\documentclass[11pt]{article}

\usepackage[margin=1in]{geometry}
\usepackage{amsmath,amssymb,amsfonts}
\usepackage{amsthm}
\usepackage{graphicx}
\usepackage{authblk}
\usepackage{xcolor}
\usepackage{placeins}
\usepackage{adjustbox}
\usepackage{authblk}
\usepackage[misc,geometry]{ifsym}
\usepackage{cite}
\usepackage{algorithm}
\usepackage[noend]{algpseudocode}

\usepackage{tikz}
\usetikzlibrary{arrows.meta,calc}

\newtheorem{theorem}{Theorem}[section]
\newtheorem{lemma}[theorem]{Theorem}
\newtheorem{observation}[theorem]{Observation}

\theoremstyle{remark}
\newtheorem{remark}[theorem]{Remark}

\newcommand{\N}{\mathbb{N}}
\newcommand{\Q}{\mathbb{Q}}
\newcommand{\R}{\mathbb{R}}

\newcommand{\cyl}[1]{[\![#1]\!]}

\DeclareMathOperator{\ent}{H}
\DeclareMathOperator{\adim}{adim}
\DeclareMathOperator{\aDim}{aDim}
\DeclareMathOperator{\odim}{odim}
\DeclareMathOperator{\oDim}{oDim}

\DeclareMathOperator{\mh}{MH}
\DeclareMathOperator{\Bern}{Bern}

\newcommand{\odimFS}[1]{\odim_{\textup{FS}}^{#1}}
\newcommand{\oDimFS}[1]{\oDim_{\textup{FS}}^{#1}}
\newcommand{\adimFS}[1]{\adim_{\textup{FS}}^{#1}}
\newcommand{\aDimFS}[1]{\aDim_{\textup{FS}}^{#1}}

\title{One Adaptive Trailing Head Can Outperform Many Oblivious Trailing Heads}

\author{Julianne Cruz}
\author{Sho Glashausser}
\author{Neil Lutz}

\affil{Swarthmore College}

\date{}
    
\begin{document}
\maketitle

\begin{abstract} 
    In the setting of multi-head finite-state dimensions, \emph{trailing heads} lag behind a \emph{leading head}, accessing past data to aid a finite-state gambler placing bets on successive bits read by the leading head. Cruz, Glashausser, Li, and Lutz (2026) proved that, for any fixed number of trailing heads, adaptive (data-dependent) movement rules can strictly outperform oblivious (data-independent) movement schedules. In this paper we strengthen that separation by proving that a \emph{single} trailing head with adaptive movements can outperform, by a large and uniform margin, arbitrarily many trailing heads with oblivious movements. Formally, our main theorem states that there is a binary sequence whose adaptive two-head finite-state strong dimension is less than its oblivious multi-head finite-state dimension, and that the gap is greater than 0.3.
\end{abstract}

\section{Introduction}

Multi-head finite-state dimensions quantify the predictability of an infinite sequence to a finite-state gambler that is granted restricted access to past data. These dimensions were introduced by Huang, Li, Lutz, and Lutz~\cite{mhfsd} as modest generalizations of finite-state dimension, which was introduced by Dai, Lathrop, Lutz, and Mayordomo~\cite{DLLM04} and is closely connected to compression, entropy rates, and normality~\cite{BoHiVi2005,AHLM07}. Unlike finite-state dimension, multi-head finite-state dimension can capture certain long-range dependencies, such as in a sequence $S$ that satisfies $S[2n]=S[n]$ for all natural numbers $n$.

The prediction is performed by a multi-head finite-state gambler that has both a \emph{leading head}, which moves continually forward through the sequence, and a set of \emph{trailing heads}, each of which can start and stop moving forward. This allows the trailing heads to lag behind the leading head and access data it has read in the distant past. The gambler has finite memory, but it updates its state based on the observations of all heads. Huang et al.~\cite{mhfsd} showed that each additional trailing head strictly increases the gambler's predictive power, and Lutz~\cite{mhfsc} showed that the same model can be characterized in terms of data compression.

In the original setting of Huang et al.~\cite{mhfsd}, the movement of the heads is \emph{oblivious} or \emph{data-independent}, meaning that the heads follow the same schedule of movement on every sequence. Cruz, Glashausser, Li, and Lutz~\cite{CGLL26} studied an \emph{adaptive} or \emph{data-dependent} version, where the gambler can use the sequence to determine which trailing heads to move forward in each step. Cruz et al.~\cite{CGLL26} showed that adaptivity strictly increases the gambler's predictive power compared to obliviousness.

Specifically, Cruz et al.~\cite{CGLL26} proved that for each $h\geq 2$, there is a sequence $X$ whose adaptive $h$-head finite-state predimension is strictly less than its oblivious $h$-head finite-state predimension, and whose $h$-head finite-state strong predimension is strictly less than its oblivious $h$-head finite-state strong predimension. Using their notation,
\begin{equation}\label{eq:cgllsep}
	\adimFS{(h)}(X)<\odimFS{(h)}(X)\text{ and }\aDimFS{(h)}(X)<\oDimFS{(h)}(X).
\end{equation}
Roughly, $\adimFS{(h)}(X)$ and $\odimFS{(h)}(X)$ quantify the infinitely-often predictability of the sequence, while the \emph{strong} variants $\aDimFS{(h)}(X)$ and $\oDimFS{(h)}(X)$ quantify its almost-everywhere predictability. Directly analogous to the hierarchy Huang et al.~\cite{mhfsd} proved in the oblivious setting, Cruz et al.~\cite{CGLL26} further proved that, for each $h\geq 2$, there is a binary sequence $Y$ such that
\begin{equation}\label{eq:cgllhier}
	\adimFS{(h)}(Y)\leq\aDimFS{(h)}(Y)<\adimFS{(h-1)}(Y)\leq\aDimFS{(h-1)}(Y).
\end{equation}

Combined, these results describe the situation one might expect: more heads help, and adaptivity helps. In this work, we consider the tradeoffs between these two types of enhancement, and we establish a much stronger separation between adaptivity and obliviousness. Namely, we prove as our main theorem that there is a binary sequence $S$ such that for all $h\geq 1$,
\begin{equation}\label{eq:cglsep}
	\adimFS{(2)}(S)\leq\aDimFS{(2)}(S)<\odimFS{(h)}(S)-0.3\leq\oDimFS{(h)}(S)-0.3.
\end{equation}
This strengthens the separations~\eqref{eq:cgllsep} from~\cite{CGLL26} in three ways. First, our separation is quantitatively much larger and uniform in $h$; the gaps in~\eqref{eq:cgllsep} proved in~\cite{CGLL26} depend on $h$ and on the size of the sequence's alphabet, but they are all less than $\frac{3}{160}=0.01875$ and converge to 0 as $h$ approaches $\infty$, decreasing on the order of $h^{-4}$. Second, we separate adaptive strong predimension from oblivious ordinary predimension. Third, and most importantly, in our separation the number of heads on the adaptive side is fixed at two. This means that an adaptive two-head finite state gambler---which has only a single trailing head with adaptive movement rules---is able to predict $S$ more successfully than \emph{any} oblivious multi-head finite-state gambler could.

We prove the separation~\ref{eq:cglsep} using a sequence $S$ that references itself and also describes its own self-referential structure. This allows an adaptive gambler to optimally position its trailing head and thereby predict every odd-indexed bit with certainty. We use a Kolmogorov complexity argument, combined with the observation that oblivious trailing heads must move at asymptotically rational speeds, to bound the performance of oblivious multi-head finite-state gamblers on the same sequence.

As we describe in Section~\ref{sec:main}, it follows from results of Huang et al.~\cite{mhfsd} and Cruz et al.~\cite{CGLL26} that for all $h\geq 2$ there is a sequence $Y$ such that $\oDimFS{(h)}(Y)<\adimFS{(h-1)}(Y)$. Combined with~\eqref{eq:cglsep}, this means the usefulness of adaptive trailing head movement is \emph{incomparable} to the usefulness of additional trailing heads; which of these enhancements adds more predictive power depends on the type of sequence being predicted. This suggests a rich structure of tradeoffs for multi-head finite-state gamblers, analogous to the nuances of the varieties of multi-head automata~\cite{HKM09,DuKrPa2020}, and we expect future work to further explore these tradeoffs.

The rest of the paper is organized as follows. In Section~\ref{sec:prelim} we briefly overview information-theoretic preliminaries, including a discussion of Kolmogorov complexity with respect to a given measure, and in Section~\ref{sec:mhfsd} we recap the definitions of oblivious and adaptive multi-head finite-state predimensions and dimensions introduced by~\cite{mhfsd} and~\cite{CGLL26}. In Section~\ref{sec:seq}, we describe the family of sequences we will use to prove our main theorem. We prove the adaptive upper bound (Theorem~\ref{thm:ub}) and oblivious lower bound (Theorem~\ref{thm:lb}) in Sections~\ref{sec:ub} and~\ref{sec:lb}, respectively, before presenting and discussing our main theorem (Theorem~\ref{thm:main}) in Section~\ref{sec:main}.

\section{Algorithmic Information Preliminaries}\label{sec:prelim}

For $m,n\in\N$, we write $[m:n]$ for the integer interval $\{m,\ldots,n-1\}$; if $m\geq n$, then this interval is empty. The space of all infinite binary sequences is $\{0,1\}^\omega$. Given any sequence $X\in\{0,1\}^\omega$ and $A\subseteq\N$, we write $X[A]$ for the sequence or string containing the bits of $X$ whose indices are in $A$, concatenated in increasing order of index. For $m,n\in\N$, we write $X[m:n]$ as shorthand for $X[[m:n]]$, so $X[m:n]=X[m]\ldots X[n-1]$; if $m\geq n$, then this is the empty sequence, denoted $\lambda$. More generally, for $A\subseteq\R$, we define $X[A]=X[A\cap\N]$.

A finite string $w\in\{0,1\}^*$ is a \emph{prefix} of a sequence $x\in\{0,1\}^\omega$, and write $w\sqsubseteq X$, if there is some sequence $Y\in\{0,1\}^\omega$ such that $X=wY$. For each finite string $w\in\{0,1\}^*$, the \emph{cylinder} of $w$ is $\cyl{w}$, the set of all binary sequences that have $w$ as a prefix:
\[\cyl{w}=\left\{X\in\{0,1\}^\omega:w\sqsubseteq X\right\}.\]

\subsection{Bernoulli Measures and Binary Entropy}
For each $p\in[0,1]$, define the \emph{$p$-biased Bernoulli distribution}
\[\Bern_p:\{0,1\}\to[0,1]\]
by $\Bern_p(1)=p$ and $\Bern_p(0)=1-p$, and define the \emph{$p$-biased Bernoulli measure} $\mu_p:\{0,1\}^\omega\to[0,1]$ by, for all strings $w\in\{0,1\}^*$,
\[\mu_p(\cyl{w})=p^{\#_1(w)}(1-p)^{\#_0(w)},\]
where $\#_1(w)$ and $\#_0(w)$ denote the number of ones in $w$ and number of zeros in $w$, respectively. Thus, $\mu_p(\cyl{w})$ is exactly the probability that a sequence of $p$-biased Bernoulli trials will begin with $w$.

The \emph{binary entropy} of $p\in[0,1]$ is the entropy of the Bernoulli distribution with success probability $p$, denoted
\begin{equation}\label{eq:H}
    \ent(p)=p\log\left(\frac{1}{p}\right)+(1-p)\log\left(\frac{1}{1-p}\right),
\end{equation}
where the logarithms, like all others in this paper, are base-2.

\subsection{$s$-Gales and Martingales}

For $s\in[0,\infty)$, an \emph{$s$-gale} on binary strings is a function $d:\{0,1\}^*\to[0,\infty)$ that satisfies, for all $w\in\{0,1\}^*$,
\[d(w)=\frac{d(w0)+d(w1)}{2^s}.\]
Informally, an $s$-gale represents the capital of a gambler betting on successive bits of some sequence $S\in\{0,1\}^\omega$, where the parameter $s$ quantifies the favorability of the betting environment. After betting on the bits in some prefix $S[0:n]$, the gambler has capital $d(S[0:n])$. The gambler then places bets on the next bit $S[n]$, allocating fraction $\beta$ of its current capital to the event $S[n]=1$ and the remaining $(1-\beta)$ fraction of its capital to the event $S[n]=0$. If $S[n]=1$, then the gambler's capital is updated to
\[d_G(S[0:n+1])=2^s\beta d_G(S[0:n]);\]
otherwise, $S[n]=0$ and its capital is updated to
\[d_G(S[0:n+1])=2^s(1-\beta) d_G(S[0:n]).\]
A \emph{martingale} is a 1-gale.

An $s$-gale $d$ \emph{succeeds} on a sequence $S\in\{0,1\}^\omega$ if
\[\limsup_{n\to\infty}d(S[0:n])=\infty,\]
and it \emph{strongly succeeds} on $S$ if
\[\liminf_{n\to\infty}d(S[0:n])=\infty.\]

\subsection{Kolmogorov Complexity}

Let $U$ be a fixed universal prefix-free Turing machine. For binary strings $x,y\in\{0,1\}^*$, the \emph{(prefix) conditional Kolmogorov complexity} of $x$ given $y$ is
\[K(x\mid y)=\min\{|z|:z\in\{0,1\}^*\text{ and }U(z,y)=x\}.\]
Intuitively, $z$ is interpreted as a program that outputs $x$ when it is given $y$ as an input. The \emph{(prefix) Kolmogorov complexity} of $x$ is the conditional Kolmogorov complexity of $x$ given the empty string:
\[K(x)=K(x\mid \lambda).\]
Our lower-bounding arguments in Section~\ref{sec:lb} will use basic properties of Kolmogorov complexity, as described in Chapter 3 of~\cite{LiVit19}. While many details of that chapter concern error terms that are logarithmic or sublogarithmic in $|x|$, our Kolmogorov complexity arguments are robust to logarithmic error and can therefore be presented more simply. For example, when $|x|,|y|\leq n$, we will use \emph{symmetry of information} in the simplified form
\begin{equation}\label{eq:soi}
    |K(x,y)-K(x\mid y)-K(y)|=O(\log n),
\end{equation}
in place of the more precise statement $|K(x,y)-K(x\mid y,K(y))-K(y)|=O(1)$.

\subsection{Martin--L\"of Randomness with Respect to a Measure}
Given any computable measure $\mu$, a sequence $X\in\{0,1\}^\omega$ is \emph{Martin-L\"of random with respect to $\mu$} (or \emph{$\mu$-Martin-L\"of random}) if, for all prefixes $w$ of $X$,
\begin{equation}\label{eq:random}
    K(w)\geq \log\left(\frac{1}{\mu(\cyl{w})}\right)-O(1).
\end{equation}
Here we interpret $\frac{1}{0}$ as $\infty$, so this condition cannot be satisfied when $\mu([\![w]\!]=0$. This Kolmogorov complexity characterization for arbitrary computable measures is due to G\'acs~\cite{Gacs1980}, generalizing the Levin--Schnorr theorem; see Theorem 4.11 of Reimann~\cite{Reimann2024}. The theory of Martin-L\"of randomness, especially with respect to the uniform measure, is covered in depth by Downey and Hirschfeldt~\cite{DowHir10}.

It is well-known that inequality~\eqref{eq:random} is almost tight, in the sense that, for every computable measure $\mu$ and every $w\in\{0,1\}^*$,
\begin{equation}\label{eq:maxcomp}
    K(w)\leq \log\left(\frac{1}{\mu(\cyl{w})}\right)+O(\log|w|).
\end{equation}
The uniform-measure case $K(w)\leq |w|+O(\log|w|)$ is given by Theorem 3.2.1 of~\cite{LiVit19}, and the generalization is straightforward via Shannon--Fano coding.

In particular, for each computable $p\in[0,1]$, if a sequence $R\in\{0,1\}^\omega$ is Martin-L\"of random with respect to the $p$-biased Bernoulli measure $\mu_p$ and $w$ is a substring of $R$ (meaning $w=R[m:n]$ for some $m,n\in\N$ with $m<n$), then
\begin{equation}\label{eq:bernoullirandom}
    \left|K(w)-\left(\#_1(w)\log\left(\frac{1}{p}\right)+\#_0(w)\log\left(\frac{1}{1-p}\right)\right)\right|=O(\log|w|).
\end{equation}
Intuitively, such a sequence $R$ is a ``typical'' outcome of an infinite sequence of $p$-biased Bernoulli trials.

It will be important for our lower bound in Section~\ref{sec:lb} that the density of ones in any prefix of a $\mu_p$-Martin-L\"of random binary sequence is close to $p$, which we quantify in the following lemma.
\begin{lemma}\label{lem:ones}
    If $p\in(0,1)$ is computable and $R\in\{0,1\}^\omega$ is a $\mu_p$-Martin-L\"of random sequence, then for every prefix $w$ of $R$,
    \[\big|\#_1(w)-p|w|\big|=O\big(\sqrt{|w|\log |w|}\big).\]
\end{lemma}
The $p=1/2$ case is immediate from Lemma 2.6.1 in the book of Li and Vit\'anyi~\cite{LiVit19}, using the deficiency function $\delta(n)=2\log n$. That result originally appeared in those authors' earlier paper~\cite{LiVit94}, generalizing Vovk~\cite{Vovk87}. The argument for arbitrary computable $p\in[0,1]$ is similar; we include it here for completeness.
\begin{proof}
    Let $p\in (0,1)$ be computable. For each $n\in\N$, let
    \[B_n=\left\{x\in\{0,1\}^n:\big|\#_1(x)-pn\big|>3\sqrt{\frac{pn\log n}{\log e}}\right\}.\]
    Letting $\cyl{B_n}$ denote $\bigcup_{x\in B_n}\cyl{x}$, a standard Chernoff bound gives
    \begin{align*}
        \mu_p(\cyl{B_n})&\leq2\exp\left(-\frac{pn}{3}\left(\frac{3}{pn}\sqrt{\frac{pn\log n}{\log e}}\right)^2\right)\\
        &=2^{1-3\log n}.
    \end{align*}
    
    Let $R\in\{0,1\}^\omega$ be a sequence, suppose there is some $n\in\N$ such that  $R[0:n]\in B_n$, and let $w=R[0:n]$. We can specify $w$ with a Shannon--Fano code (see Section 5.4 of~\cite{CovTho06}) for $\mu_p$ restricted to $B_n$. This code can be computed given $n$, and $w$ will have codeword length
    \[\left\lceil\log\left(\frac{\mu_p(\cyl{B_n})}{\mu_p(\cyl{w})}\right)\right\rceil\leq 2-3\log (n)+\log\left(\frac{1}{\mu_p(\cyl{w})}\right).\]
    Therefore,
    \begin{align*}
        K(w)&\leq K(n)+2-3\log(n)+\log\left(\frac{1}{\mu_p(\cyl{w})}\right)\\
        &\leq O(1)-\log(n)+\log\left(\frac{1}{\mu_p(\cyl{w})}\right),
    \end{align*}
    because $K(n)\leq \log(n)+2\log\log(n)+O(1)\leq 2\log(n)+O(1)$; see Proposition 3.6.3 of~\cite{DowHir10}. This violates~\eqref{eq:random} whenever $n$ is sufficiently large. As the lemma statement holds trivially for bounded $|w|$, this completes the proof.
\end{proof}

By combining Lemma~\ref{lem:ones} with~\eqref{eq:bernoullirandom}, we have that if $R\in\{0,1\}^\omega$ is $\mu_p$-Martin-L\"of random and $w$ is a substring of $R$, then
\begin{equation}\label{eq:bernoullirandoment}
    \big|K(w)-\ent(p)|w|\big|=O\big(\sqrt{|w|\log|w|}\big).
\end{equation}

\section{Multi-Head Finite-State Dimensions}\label{sec:mhfsd}

We now recall the definitions of multi-head finite-state gamblers, predimensions, and dimensions. The oblivious variants were introduced first, by Huang et al.~\cite{mhfsd}, but here we treat oblivious gamblers as a special case of the more general adaptive gamblers studied by Cruz et al.~\cite{CGLL26}. 

\subsection{Multi-Head Finite-State Gamblers}
For $h\geq 1$, an \emph{(adaptive) $h$-head finite-state gambler} or \emph{$h$-FSG} is a 6-tuple
\[G=(Q,\Sigma,\delta,\beta,q_0,c_0),\]
where $Q$ is the finite state space, $\Sigma$ is a finite alphabet of size $\geq 2$,
\[\delta:Q\times\Sigma^h\to Q\times\{0,1\}^{h-1}\]
is the transition function, $\beta:Q\to\Delta_\Q(\Sigma)$ is the betting function, $q_0\in Q$ is the initial state, and $c_0\in[0,\infty)$ is the initial capital. Here $\Delta_\Q(\Sigma)$ denotes the class of all rational-valued discrete probability distributions over $\Sigma$. For the remainder of this paper, the alphabet $\Sigma$ will always be $\{0,1\}$.

The gambler has $h$ \emph{heads}, numbered $1,\ldots,h$. Head $h$ is the \emph{leading head}, and the others are \emph{trailing heads}. Given a sequence $S\in\{0,1\}^\omega$, the gambler proceeds in discrete time steps. At each step $n\in\N$, each head $i$ has position $\pi_i(S[0:n])$, and the state is $q(S[0:n])$. Initially, the state is $q(S[0:0])=q_0$ and all heads are at position 0, meaning $\pi_i(S[0:0])=0$ for $1\leq i\leq h$.

During step $n\in\N$, in some state $q\in Q$, the gambler $G$ does the following. First, $G$ places a bet $\beta(q)$ on bit $S[n]$. Second, $G$ applies its transition function to its current state and the $h$-tuple of bits at the current positions of each head, yielding a new state $q'$ and an $(h-1)$-tuple $(r_1,\ldots,r_{h-1})\in\{0,1\}^{h-1}$:
\[\delta(q(S[0:n]),(S[\pi_1(S[0:n])],\ldots,S[\pi_h(S[0:n])]))=(q',(r_1,\ldots,r_{h-1})).\]
Third, $G$ updates its state to $q(S[0:n+1])=q'$ and, defining $r_h=1$, moves each head $i$ forward if $r_i=1$. That is, for each $i\in\{1,\ldots,h\}$,
\[\pi_i(S[0:n+1])=\pi_i(S[0:n])+r_i.\]
In particular, $\pi_h(S[0:n])=n$ holds for all $n\in\N$; the leading head moves forward in every step.

\subsection{Obliviousness and Adaptivity}
An $h$-head finite-state gambler is \emph{oblivious} or \emph{data-independent} if its trailing head movements are independent of the input sequence, meaning that for all sequences $R,S\in\{0,1\}^\omega$, all $n\in\N$, and all $i\in\{1,\ldots,h\}$, we have
\[\pi_i(R[0:n])=\pi_i(S[0:n]).\]
Formally, oblivious multi-head finite-state gamblers are a special case of adaptive $h$-head finite-state gamblers, but we primarily use the adjective \emph{adaptive} to contrast the general model with the oblivious special case.

Huang et al.~\cite{mhfsd} observed that as a consequence of oblivious, finite-state movement, each trailing head $i$ must have a fixed speed $\sigma_i\in[0,1]$ such that, whenever the position of the leading head is $n$, the position of head $i$ is within a constant of $\sigma_i$. It is implicit in the proof of this observation in~\cite{mhfsd} that $\sigma_i$ must be rational, as it is the integer number of times the head advances in a particular cycle, divided by the integer length of the cycle.
\begin{observation}[Huang et al. \cite{mhfsd}]\label{obs:speed}
    If $h\geq 2$ and $G$ is an oblivious $h$-head finite-state gambler, then there are speeds $\sigma_1,\ldots,\sigma_{h-1}\in[0,1]\cap\Q$ and a constant $c_G$ such that, for all $S\in\{0,1\}^\omega$, all $n\in\N$, and all $i\in\{1,\ldots,h-1\}$,
    \[\sigma_i n-c_G\leq \pi_i(S[0:n])\leq \sigma_i n+c_G.\]
\end{observation}
It follows that a trailing head with speed 0 or 1 can be simulated by the leading head using finite states, so we will assume $\sigma_i\in(0,1)$ for $1\leq i\leq h-1$. The rationality of each $\sigma_i$ is a key ingredient in the current paper; our adaptive gambler's trailing head will move at an asymptotically irrational speed, thereby accessing information that is unavailable to heads moving at rational speeds.

\subsection{Varieties of Multi-Head Finite-State Dimension}
    The \emph{martingale of a multi-head finite-state gambler} is the function $d_G:\{0,1\}^*\to[0,\infty)$ defined recursively by $d_G(S[0:0])=c_0$ and, for all $n\in\N$,
    \begin{equation}\label{eq:dG}
        d_G(S[0:n+1])=2\beta(q(S[0:n]))(S[n])d_G(S[0:n])
    \end{equation}
    and the \emph{$s$-gale of $G$} is the function $d_G^{(s)}:\{0,1\}^*\to[0,\infty)$ given by
    \begin{equation}\label{eq:dGs}
        d_G^{(s)}(S[0:n])=2^{(s-1)n}d_G(S[0:n]).
    \end{equation}

    For each $h\geq 1$ and $S\in\{0,1\}^\omega$, the \emph{adaptive $h$-head finite-state predimension}, \emph{adaptive $h$-head finite-state strong predimension}, \emph{oblivious $h$-head finite-state predimension}, and \emph{oblivious $h$-head finite-state strong predimension} of $S$ are, respectively:
    \begin{align*}
        \adimFS{(h)}(S)&=\inf\{s:\exists\text{ $h$-FSG $G$ s.t. $d_G^{(s)}$ succeeds on }S\},\\
        \aDimFS{(h)}(S)&=\inf\{s:\exists\text{ $h$-FSG $G$ s.t. $d_G^{(s)}$ succeeds strongly on }S\},\\
        \odimFS{(h)}(S)&=\inf\{s:\exists\text{ oblivious $h$-FSG $G$ s.t. $d_G^{(s)}$ succeeds on }S\},\\
        \oDimFS{(h)}(S)&=\inf\{s:\exists\text{ oblivious $h$-FSG $G$ s.t. $d_G^{(s)}$ succeeds strongly on }S\}.
    \end{align*}
    These quantities directly generalize the one-head cases, which are finite-state dimension as defined by Dai et al.~\cite{DLLM04} and finite-state strong dimension as defined by Athreya, Hitchcock, Lutz, and Mayordomo~\cite{AHLM07}. Obliviousness and adaptivity do not come into play in the one-head case since these properties only concern the trailing heads.
    
    It is immediate from the above definitions that, for each $S\in\{0,1\}^\omega$,
    \begin{equation}\label{eq:comparevariants}
        \begin{aligned}
            \adimFS{(h)}(S)&\leq \aDimFS{(h)}(S)\leq \oDimFS{(h)}(S),\\
            \adimFS{(h)}(S)&\leq \odimFS{(h)}(S)\leq \oDimFS{(h)}(S),
        \end{aligned}
    \end{equation}
    and each of these quantities is non-increasing in $h$.
    
    The \emph{oblivious multi-head finite state dimension}, \emph{oblivious multi-head finite-state strong dimension}, \emph{adaptive multi-head finite state dimension}, and \emph{adaptive multi-head finite-state strong dimension} of $S$ are, respectively, the limits inferior, over all positive integers $h$, of the four quantities above.

\section{Constructing the Family of Sequences}\label{sec:seq}

Define the function $F:\{0,1\}^\omega\to\{0,1\}^\omega$ by, for all sequences $R\in\{0,1\}^\omega$, for all $n\in\N$,
\begin{equation}\label{eq:F}
    \begin{aligned}
        F(R)[2n]&=R[n]\\
        F(R)[2n+1]&=F(R)\left[2\#_1(R[0:n])\right].
    \end{aligned}
\end{equation}
Thus, at even indices, $F(R)$ takes a new bit directly from $R$, and at odd indices, it copies a bit from earlier in the sequence $F(R)$. Note that this earlier bit is always even-indexed, so the recursive structure only has depth one. The exact index of the copied bit depends on how many ones appeared in $R$ prior to $R[n]$, which is the bit most recently taken from $R$.

\begin{figure}
    \centering
    \begin{adjustbox}{width=\textwidth}
\begin{tikzpicture}[
  >=Latex,
  bit/.style={draw, minimum width=6mm, minimum height=6mm, inner sep=0pt, font=\ttfamily},
  idx/.style={font=\scriptsize},
  rarrow/.style={->, semithick},
  copyarrow/.style={->, semithick}
]

\def\Rbits{1,0,1,0,0,1,1,0,1,1,1,0,1,1,0,1,0,1,1,0}
\def\Fbits{1,1,0,0,1,0,0,1,0,1,1,1,1,0,0,0,1,0,1,1}

\foreach \i in {0,...,19}{
  \node[idx] at (0.6*\i,1.5) {\i};
}

\node[anchor=east] at (-0.8,0.8) {$R$};
\node[anchor=east] at (-0.8,-1.2) {$F(R)$};

\foreach[count=\i from 0] \b in \Rbits {
  \node[bit] (R\i) at (0.6*\i,0.8) {\b};
}

\foreach[count=\i from 0] \b in \Fbits {
  \node[bit] (F\i) at (0.6*\i,-1.2) {\b};
}

\foreach \i/\j in {0/0,1/2,2/4,3/6,4/8,5/10,6/12,7/14,8/16,9/18}{
  \draw[rarrow] (R\i.south) -- (F\j.north);
}

\draw[copyarrow] (F0.south)  to[out=-60,in=-120] (F1.south);

\draw[copyarrow] (F2.south)  to[out=-60,in=-120] (F3.south);
\draw[copyarrow] (F2.south)  to[out=-60,in=-120] (F5.south);

\draw[copyarrow] (F4.south)  to[out=-60,in=-120] (F7.south);
\draw[copyarrow] (F4.south)  to[out=-60,in=-120] (F9.south);
\draw[copyarrow] (F4.south)  to[out=-60,in=-120] (F11.south);

\draw[copyarrow] (F6.south)  to[out=-60,in=-120] (F13.south);

\draw[copyarrow] (F8.south)  to[out=-60,in=-120] (F15.south);
\draw[copyarrow] (F8.south)  to[out=-60,in=-120] (F17.south);

\draw[copyarrow] (F10.south) to[out=-60,in=-120] (F19.south);

\end{tikzpicture}
\end{adjustbox}
    \vspace{-2em}
    \caption{Example of applying $F$ to the first ten bits of a sequence $S$. Each bit of $S$ is placed at an even-indexed position in $F(R)$, and each of these bits is copied forward to one or more odd-indexed positions in $F(R)$.}
    \label{fig:example}
\end{figure}
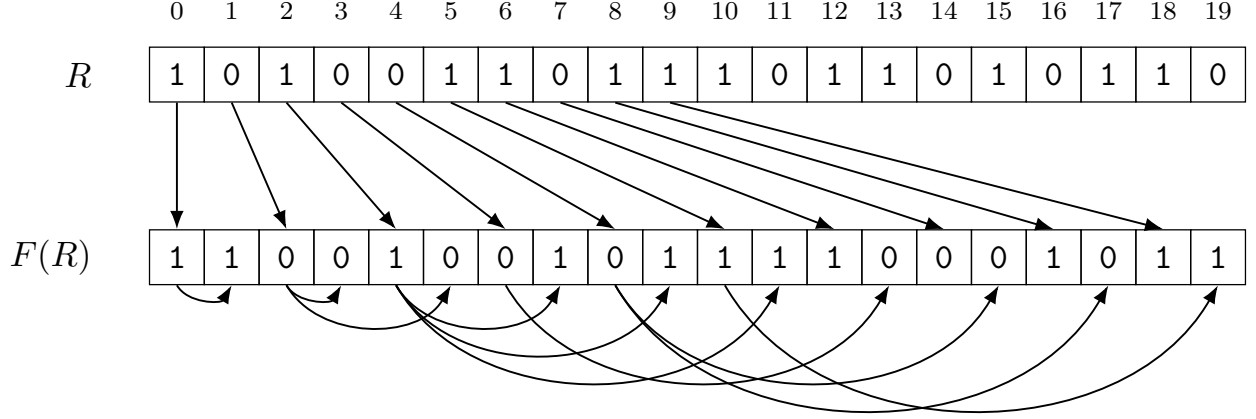

For $p\in[0,1]$, we will apply the function $F$ to $\mu_p$-Martin-L\"of random sequences. For such a sequence $R$, Lemma~\ref{lem:ones} tells us that $\#_1(R[0:n])$ will be approximately $pn$, so $F(R)[2n+1]$ will copy a bit from an earlier in $F(R)$ that is close to index $2pn$.

\section{Upper Bound on Adaptive Two-Head Strong Predimension}\label{sec:ub}

We now prove our upper-bound on the adaptive two-head strong finite-state predimension of sequences of the form $F(R)$, where $R$ is random with respect to a Bernoulli measure $\mu_p$, by defining an adaptive two-head finite state gambler whose $s$-gale strongly succeeds on $F(R)$ for all $s>\ent(p)/2$. By always positioning its trailing head near the currently pertinent information, the gambler will predict odd-indexed bits with certainty. It will treat the even-indexed bits as $\tilde{p}$-biased Bernoulli trials, where $\tilde{p}$ is a rational number very close to $p$. This small complication is due to the restriction that the betting function of a multi-head finite-state gambler is a rational-valued distribution.

\begin{theorem}\label{thm:ub}
    If $p\in(0,1)$ and $R\in\{0,1\}^\omega$ is a $\mu_p$-Martin-L\"of random sequence, then
    \[\aDimFS{(2)}(F(R))\leq \frac{\ent(p)}{2}.\]
\end{theorem}

For each $p\in[0,1]\cap\Q$, define the adaptive $2$-head finite-state gambler
\[G_p = (Q,\{0,1\},\delta,\beta_p,q_0,1)\]
as follows.
\begin{itemize}
    \item The state space is $Q = \{0, 1 \}^3$; in each state $(b_1,b_2,b_3)$, $b_1$ stores the last bit read by the trailing head, $b_2$ stores the last \emph{even-indexed} bit read by the leading head, and $b_3$ indicates the parity of the leading head's position.
    \item The transition function $\delta:Q\times\{0,1\}^2\to Q\times\{0,1\}$ is given by
    \[\delta((b_1, b_2, b_3), (a_1, a_2)) =
    \begin{cases}
        ((a_1,a_2,1), a_2)&\text{if }b_3=0\\
        ((a_1,b_2,0), b_2)&\text{if }b_3=1.
    \end{cases}.\]
    Informally, the trailing head moves forward twice each time the leading head reads a one at an even index. Since it can move forward at most once per step, this is implemented by the trailing head moving forward immediately when the one is read by the leading head at an even index, then again in the next step, when the leading head is at an odd index. Thus, the trailing head moves when $b_3=0$ and $a_2=1$, and when $b_3=1$ and $b_2=1$, which are exactly the conditions for the last entry in $\delta((b_1, b_2, b_3), (a_1, a_2))$ to be 1.
    \item The betting function $\beta_p:Q\to\Delta_\Q(\{0,1\})$ is given by
    \[
        \beta_p((b_1, b_2, b_3)) =
        \begin{cases}
            \Bern_{b_1} & \text{if } b_3 = 1\\
            \Bern_p & \text{otherwise.}
        \end{cases}
    \]
    Note that $\Bern_{b_1}$ is a point distribution at 0 or 1, which informally means that $G_p$ has total confidence when predicting the odd-indexed bits.
    \item The initial state is $q_0 = (0, 0, 0)$.
\end{itemize}

Now let $p\in(0,1)$, let $R\in\{0,1\}^\omega$ be Martin-L\"of random with respect to $\mu_p$, and let $S=F(R)$. To prove Theorem~\ref{thm:ub}, it suffices to show that for all $s>\frac{\ent(p)}{2}$, there is an adaptive two-head finite-state gambler whose $s$-gale strongly succeeds on $S$. For this, let $\varepsilon\in (0,\min\{p,1-p\}/2)$,
and let $\tilde{p}\in\left(p-\varepsilon,p+\varepsilon\right)\cap\Q$
satisfy $|\ent(p)-\ent(\tilde{p})|<\varepsilon;$ such a $\tilde{p}$ exists because $\ent$ is continuous. Let
\begin{equation}\label{eq:s}
    s=\frac{\ent(p)}{2}+\left(3-\log(p)-\log(1-p)\right)\varepsilon.
\end{equation}

\begin{lemma}\label{lem:Gtildep}
    The $s$-gale $d_{G_{\tilde{p}}}^{(s)}$ strongly succeeds on $S$.
\end{lemma}

\begin{proof}    
    We first show that $G_{\tilde{p}}$ correctly predicts every odd-indexed bit of $S$. For each $n\in\N$, let $q(n)=(b_1(n),b_2(n),b_3(n))\in Q$ be the state of the $G_{\tilde{p}}$ when it bets on the bit $S[n]$, i.e., $q(n)=q(S[0:n])$.
    
    Consider the state
    \[q(2n+1)=(b_1(2n+1),b_2(2n+1),b_3(2n+1))\]
    of $G_{\tilde{p}}$ when it bets on a bit $S[2n + 1]$, for some arbitrary $n \in \N$. Since the parity bit is initially 0 and flips in each step, we will have $b_3(2n+1)=1$, and $\beta_{\tilde{p}}(q(2n+1))=\Bern_{b_1(2n+1)}$. As described in the gambler construction, $b_1(2n+1)$ stores the last bit read by the trailing head, and the trailing head moves forward twice each time the leading head read a one at an even index of $S$. Thus, the position of the trailing head is twice the number of ones that appear at even indices in $S[0:2n]$. By the definition of $F$, this is exactly $\#_1(R[0:n])$, so
    \[b_1(2n+1)=S[2\#_1(R[0:n])]=S[2n+1],\]
    and it follows that
    \[\beta_{\tilde{p}}(q(2n+1))(S[2n+1])=\Bern_{S[2n+1]}(S[2n+1])=1.\]

    At even-indexed bits $S[2n]$, for some $n\in\N$, the parity bit in the state $q(2n)$ of $G_{\tilde{p}}$ is 0, so the betting distribution is $\beta_{\tilde{p}}(q(2n))=\Bern_{\tilde{p}}$.

    We now use these betting distributions to prove that for all sufficiently large $n\in\N$,
    \begin{equation}\label{eq:sgalebd}
        d_{G_{\tilde{p}}}(S[0:2n])>2^{(2-\ent(p)+(2\log(p)+2\log(1-p)-5)\varepsilon)n}.
    \end{equation}

    Recall that we have
    \begin{equation}\label{eq:epsilon}
        \varepsilon\in \left(0,\frac{\min\{p,1-p\}}{2}\right),
    \end{equation}
    and
    \begin{equation}\label{eq:tildep}
        \tilde{p}\in\left(p-\varepsilon,p+\varepsilon\right)\cap\Q,
    \end{equation}
    satisfying
    \begin{equation}\label{eq:entbd}
        |\ent(p)-\ent(\tilde{p})|<\varepsilon.
    \end{equation}
    Note that~\eqref{eq:epsilon} and~\eqref{eq:tildep} guarantee
    \begin{equation}\label{eq:half}
        \tilde{p}>\frac{p}{2}\qquad\text{and}\qquad 1-\tilde{p}>\frac{1-p}{2}.
    \end{equation}
    
    We also have, for all $n\in\N$,
    \begin{equation}\label{eq:oddbits}
        \beta_{\tilde{p}}(q(2n+1))(S[2n+1])=\Bern_{S[2n+1]}(S[2n+1])=1
    \end{equation}
    and
    \begin{equation}\label{eq:evenbits}
        \beta_{\tilde{p}}(q(2n))=\Bern_{\tilde{p}}.
    \end{equation}
    
    Hence, for all $n\in\N$,
    \begin{align*}
        d_{G_{\tilde{p}}}(S[0:2n])&=\prod_{i=0}^{2n-1}2\beta_{\tilde{p}}(q(i))(S[i])\tag{by~\eqref{eq:dG}}\\
        &=2^{2n}\prod_{j=0}^{n-1}\Bern_{\tilde{p}}(S[2j])\tag{by~\eqref{eq:oddbits} and~\eqref{eq:evenbits}}\\
        &=2^{2n}\prod_{j=0}^{n-1}\Bern_{\tilde{p}}(R[j])\tag{by~\eqref{eq:F}}\\
        &=2^{2n}\tilde{p}^{\#_1(R[0:n])}(1-\tilde{p})^{\#_0(R[0:n])},
    \end{align*}
    and
    \begin{equation}\label{eq:evenloss}
        d_{G_{\tilde{p}}}(S[0:2n+1])\geq\min\{\tilde{p},1-\tilde{p}\}d_{G_{\tilde{p}}}(S[0:2n]).
    \end{equation}
    
    By Lemma~\ref{lem:ones}, then,
    \begin{align*}
        d_{G_{\tilde{p}}}(S[0:2n])&\geq 2^{2n}\tilde{p}^{pn+O(\sqrt{n\log n})}(1-\tilde{p})^{(1-p)n+O(\sqrt{n\log n})}\\
        &\geq 2^{2n}\tilde{p}^{pn+\varepsilon n}(1-\tilde{p})^{(1-p)n+\varepsilon n},
    \end{align*}
    for all sufficiently large $n\in\N$. We can rewrite this as
    \begin{align*}
        d_{G_{\tilde{p}}}(S[0:2n])&\geq 2^{2n}\tilde{p}^{(p+\varepsilon)n}(1-\tilde{p})^{(1-(p-\varepsilon))n}\\
        &> 2^{2n}\tilde{p}^{(\tilde{p}+2\varepsilon)n}(1-\tilde{p})^{(1-(\tilde{p}-2\varepsilon))n}\tag{by~\eqref{eq:tildep}}\\
        &=2^{2n}\tilde{p}^{\tilde{p}n+2\varepsilon n}(1-\tilde{p})^{(1-\tilde{p})n+2\varepsilon n}\\
        &=2^{2n-\ent(\tilde{p})n}\tilde{p}^{2\varepsilon n}(1-\tilde{p})^{2\varepsilon n}\tag{by~\eqref{eq:H}}\\
        &=2^{(2-\ent(\tilde{p})+2\log(\tilde{p})\varepsilon+2\log(1-\tilde{p})\varepsilon)n}\\
        &> 2^{(2-\ent(p)-\varepsilon+2\log(\tilde{p})\varepsilon+2\log(1-\tilde{p})\varepsilon)n}\tag{by~\eqref{eq:entbd}}\\
        &> 2^{(2-\ent(p)-\varepsilon+2(\log(p)-1)\varepsilon+2(\log(1-p)-1)\varepsilon)n}\tag{by~\eqref{eq:half}}\\
        &> 2^{(2-\ent(p)+(2\log(p)+2\log(1-p)-5)\varepsilon)n},
    \end{align*}
    so~\eqref{eq:sgalebd} holds for all sufficiently large $n\in\N$.
    
    Therefore, recalling~\eqref{eq:dGs} and~\eqref{eq:s}, for all sufficiently large $n\in\N$,
    \begin{align*}
        d_{G_{\tilde{p}}}^{(s)}(S[0:2n])&=2^{(\ent(p)+(6-2\log(p)-2\log(1-p))\varepsilon-2)n} d_{G_{\tilde{p}}}(S[0:2n])\\
        &>2^{\varepsilon n},
    \end{align*}
    which implies
    \[\liminf_{n\to\infty}d_{G_{\tilde{p}}}^{(s)}(S[0:2n])=\infty.\]
    At even-indexed bits, the $s$-gale value can only be multiplied by $2^s\tilde{p}$ or $2^s(1-\tilde{p})$, so
    \[d_{G_{\tilde{p}}}^{(s)}(S[0:2n+1])\geq2^s\min\{\tilde{p},(1-\tilde{p})\}\cdot d_{G_{\tilde{p}}}^{(s)}(S[0:2n])\]
    holds for all $n\in\N$.
    We therefore also have
    \[\liminf_{n\to\infty}d_{G_{\tilde{p}}}^{(s)}(S[0:n])=\infty,\]
    which is the definition of strong success on $S$.
\end{proof}

As $G_{\tilde{p}}$ was an adaptive two-head finite-state gambler and $s$ can be made arbitrarily close to $\frac{\ent(p)}{2}$ by letting $\varepsilon$ approach 0, it follows from Lemma~\ref{lem:Gtildep} that
\[\aDimFS{(2)}(F(R))\leq\frac{\ent(p)}{2}\]
completing the proof of Theorem~\ref{thm:ub}.

\section{Lower Bound on Oblivious Multi-Head Dimension}\label{sec:lb}

The trailing head movement of the adaptive gamblers described in Section~\ref{sec:ub} depend on the sequence $R$; its asymptotic speed is exactly the asymptotic density $p$ of ones in the underlying sequence $R$. We now prove that when $p$ is irrational, no oblivious multi-head finite-state gambler can match the performance of an adaptive gambler on $F(R)$.
\begin{theorem}\label{thm:lb}
    If $p\in(0,1)\setminus\Q$ is computable and $R\in\{0,1\}^\omega$ is a $\mu_p$-Martin-L\"of random sequence, then
    \[\odimFS{\mh}(F(R))\geq (1+p)\frac{\ent(p)}{2}.\]
\end{theorem}

To prove this theorem, let $p\in[0,1]\setminus \Q$ be computable, $R\in\{0,1\}^\omega$ a $\mu_p$-Martin-L\"of random sequence, $h\geq 1$, and $S=F(R)$. By the definition of $\odimFS{\mh}$, it suffices to show that
\[\odimFS{(h)}(S)\geq (1+p)\frac{\ent(p)}{2}.\]
For this $h$-head finite-state predimension lower bound, we will follow the same overall approach taken by Huang et al.~\cite{mhfsd} and Cruz et al.~\cite{CGLL26}. That is, we will prove conditional Kolmogorov complexity lower bounds on substrings $S[m:n]$, where the string being conditioned on includes all information that the trailing heads might have accessed while the leading head reads $S[m:n]$. From this lower bound, we will use the following lemma to infer an upper bound on the capital growth of $G$'s martingale during this period, from which we will then derive a lower bound on the string's $h$-head finite-state predimension.
\begin{lemma}[Cruz et al.~\cite{CGLL26}, generalizing Huang et al.~\cite{mhfsd}]\label{lem:dtok}
    Let $h\geq 1$, $G$ any $h$-head finite-state gambler, $\alpha,\varepsilon\in(0,1)\cap\Q$, and $S\in\{0,1\}^\omega$. If $n,m\in\N$ and $U\subseteq\N$ satisfy 
    \begin{enumerate}
        \item $n-m$ is sufficiently large,
        \item $U$ is uniformly computable given $n$,
        \item $\displaystyle [0:m]\cap\bigcup_{i=1}^{h-1}[\pi_i(S[0:m]):\pi_i(S[0:n])+1]\subseteq U$, and
        \item $K(S[m:n]\mid S[U])\ge (1-\alpha)(n-m)$,
    \end{enumerate}
    then
    \begin{equation}\label{eq:nowin}
        \max_{m< k\leq n}\frac{d_G(S[0:k])}{d_G(S[0:m])}\leq 2^{(\alpha+\varepsilon)(n-m)}.
    \end{equation}
\end{lemma}

\begin{remark}
    Our statement of this lemma differs slightly from~\cite{CGLL26}, in that we require $U$ to be uniformly computable. This is because in~\cite{CGLL26}, for $A\subseteq[0:n]$, $X[A]$ denoted the length-$n$ string ``masked'' string where $X[A][i]=X[i]$ if $i\in A$ and 0 otherwise, whereas in the present paper it denotes the length-$|A|$ string of bits in $X$ whose indices are in $A$. When the set $A$ is uniformly computable given $n$, these two strings are computable from each other, so their complexities can only differ by $O(\log n)$. The $\varepsilon$ error parameter in Lemma~\ref{lem:dtok} accommodates this difference whenever $n-m$ is sufficiently large, so our statement is equivalent to that of~\cite{CGLL26}.
\end{remark}

To prove the conditional Kolmogorov complexity lower bound that satisfy the fourth condition in Lemma~\ref{lem:dtok}, we will consider substrings of the form $S[2m:2n]$, to make the endpoints conveniently even. We will argue that each such substring contains almost all the information from two substrings of $R$: Speaking informally, the even-indexed bits in $S[2m:2n]$ come from $R[m:n]$. As the density of ones in $R$ is approximately $p$, the odd-indexed bits in $S[2m:2n]$ will copied forward from approximately $S[2pm:2pn]$, and these bits ultimately come from approximately $R[pm:pn]$. Thus, $S[2m:2n]$ contains all the information in approximately $(n-m)(1+p)$ bits of $R$, meaning its Kolmogorov complexity cannot be much less than $(n-m)(1+p)\ent(p)$.

Furthermore, we argue that this approximate bound holds (for appropriate choices of $m$ and $n$) even when we condition on the information provided by the trailing heads. This is because, for sufficiently large $n$ and $m$ sufficiently close to $n$, oblivious trailing heads with rational speeds can't access $S[2pm:2pn]$ while the leading head reads $S[m:n]$. All information they access during this period is about other parts of $R$, which are not informative about the pertinent substrings $R[pm:pn]$ and $R[m:n]$.

To formalize this intuitive argument, let $G$ be an oblivious $h$-head finite-state gambler. Let $\sigma_1,\ldots,\sigma_{h-1}$ be the rational trailing head speeds of $G$, as in Observation~\ref{obs:speed}, noting that $p$, being irrational, is distinct from each of these speeds. Thus, each $|\sigma_i n-pn|$ is linear in $n$. We therefore have the following.
\begin{observation}\label{obs:gap}
    There exist $\gamma_0\in(0,1)$ and $n_0\in\N$ such that, for all $\gamma\in(\gamma_0,1)$, all $n\geq n_0$, and all trailing heads $i$,
    \[[\lfloor \sigma_i\lfloor\gamma n\rfloor\rfloor-c_G,\lceil\sigma_i n\rceil+c_G]\cap[\lfloor p\lfloor\gamma n\rfloor\rfloor,\lceil pn\rceil]=\emptyset,\]
    where $c_G$ is the constant from Observation~\ref{obs:speed}
\end{observation}

Fix $\gamma_0,n_0$ as in Observation~\ref{obs:gap}, and choose rational constants $\gamma\in(\gamma_0,1)$, $\varepsilon\in (0,1)$, and
\[\alpha\in\left(1-(1+p)\frac{\ent(p)}{2},1-(1+p)\frac{\ent(p)}{2}+\varepsilon\right).\]
Let $n\geq n_0$ and $m=\lfloor\gamma n\rfloor$ be sufficiently large to satisfy the first condition of Lemma~\ref{lem:dtok} with this choice of $\varepsilon$ and $\alpha$. Let
\[U=[0:2m]\cap\bigcup_{i=1}^{h-1}[2\sigma_i m-c_G,2\sigma_in+c_G],\]
which satisfies the second condition of Lemma~\ref{lem:dtok} and, by Observation~\ref{obs:speed}, also satisfies the third condition. Then by Observation~\ref{obs:gap}, we have
\begin{equation}\label{eq:disjoint}
    U\cap [\lfloor 2pm\rfloor,\lceil 2pn\rceil]=\emptyset.
\end{equation}

Let $A=[\lfloor pm\rfloor :\lfloor pn\rfloor]\cup[m:n]$ and $B=[0:\lfloor pm\rfloor]\cup[\lfloor pn\rfloor:m]$,
noting that $A$ and $B$ partition $[0:n]$. Roughly, $R[A]$ is the part of $R$ that pertains to $S[2m:2n]$.

We now prove technical lemmas to assist in the proof of Theorem~\ref{thm:lb}. Briefly, Lemma~\ref{lem:ra} tells us that $S[2m:2n]$ contains almost all the information in $R[A]$, and Lemma~\ref{lem:rb} tells us that $R[B]$ contains almost all the information in $S[U]$, which is approximately the part of $S$ read by the trailing heads while the leading head reads $S[2m:2n]$. Since $R[A]$ and $R[B]$ are disjoint parts of a $\mu_p$-Martin-L\"of random sequence, they are essentially independent of each other, so these two lemmas imply that $S[U]$ contains almost no information about $S[m:n]$, yielding Lemma~\ref{lem:fourthcond}.

\begin{lemma}\label{lem:ra}
    $K(R[A]\mid S[2m:2n])=O\big(\sqrt{n\log n}\big)$.
\end{lemma}
\begin{proof}
    By the definition of $F$, the even-indexed bits of $S[2m:2n]$ are exactly $R[m:n]$. Furthermore, given $S[2m:2n]$, we can computationally recover
    \[R[\#_1(R[0:m]):\#_1(R[0:n])]\]
    using Algorithm~\ref{alg:StoR}.
    \begin{algorithm}
    \caption{Computing $S[2m:2n]\mapsto R[\#_1(R[0:m]):\#_1(R[0:n])]$}
    \label{alg:StoR}
    \begin{algorithmic}
    \State $w \gets \lambda$
    \For{$k = m$ \textbf{to} $n-1$}
        \If{$S[2k] = 1$}
            \State $w \gets w S[2k+1]$
        \EndIf
    \EndFor
    \State \Return $w$
    \end{algorithmic}
    \end{algorithm}
    
    By Lemma~\ref{lem:ones}, the symmetric difference
    \[[\lfloor pm\rfloor:\lfloor pn\rfloor]\:\triangle\:[\#_1(R[0:m]):\#_1(R[0:n])]\]
    has cardinality $O\big(\sqrt{n\log n}\big)$. Therefore, there is a computational process that, given $R[\#_1(R[0:m]):\#_1(R[0:n])]$ and $O\big(\sqrt{n\log n}\big)$ bits of additional information, could output $R[\lfloor pm\rfloor :\lfloor pn\rfloor]$. By combining this with Algorithm~\ref{alg:StoR}, we have
    \[K(R[A]\mid S[2m:2n])=O\big(\sqrt{n\log n}\big).\qedhere\]
\end{proof}

\begin{lemma}\label{lem:rb}
    $K(S[U]\mid R[B])=O\big(\sqrt{n\log n}\big)$.
\end{lemma}
\begin{proof}
    For any even index $2k\in U$ we have $k<m$, and~\eqref{eq:disjoint} tells us that
    \[k\not\in [\lfloor pm\rfloor:\lfloor pn\rfloor],\]
    so $S[2k]=R[k]$ for some $k\in B$. For every odd index $2k+1\in U$, Lemma~\ref{lem:ones} tells us that
    \[\#_1(R[0:k])\leq pk+O\big(\sqrt{n\log n}\big),\]
    which is less than $pm$ except for $O\big(\sqrt{n\log n}\big)$ odd indices $2k+1<2m$. Thus, with $O\big(\sqrt{n\log n}\big)$ exceptions, $S[2k+1]=R[\ell]$ for some $\ell\in B$. The non-exceptional bits come from $2(h-1)=O(1)$ intervals of indices, and the endpoints of each interval can be specified in $O(\log n)$ bits. Therefore, with $O\big(\sqrt{n\log n}\big)$ bits of side information to describe the intervals and the exceptional bits, $S[U]$ can be computed from $R[B]$.
\end{proof}

\begin{lemma}\label{lem:fourthcond}
     $K(S[2m:2n]\mid S[U])\geq \ent(p)(1+p)(n-m)-O\big(\sqrt{n\log n}\big)$.
\end{lemma}
\begin{proof}
    Applying basic properties of Kolmogorov complexity, we have
    \begin{align*}
        K(S[2m:2n]\mid S[U])&\geq K(R[A]\mid S[U])-O\big(\sqrt{n\log n}\big)\tag{by Lemma~\ref{lem:ra}}\\
        &\geq K(R[A]\mid R[B])-O\big(\sqrt{n\log n}\big)\tag{by Lemma~\ref{lem:rb}}\\
        &\geq K(R[0:n])-K(R[B])-O\big(\sqrt{n\log n}\big)\tag{by~\eqref{eq:soi}}\\
        &\geq \ent(p)n-K(R[B])-O\big(\sqrt{n\log n}\big).\tag{by~\eqref{eq:bernoullirandoment}}
    \end{align*}
    Since $R[B]$ is just the concatenation of $R[0:\lfloor pm\rfloor]$ and $R[\lfloor pn\rfloor:m]$, we have
    \begin{align*}
        K(R[B])&\leq K(R[0:\lfloor pm\rfloor])+K(R[\lfloor pn\rfloor:m])+O(1)\\
        &\leq \ent(p)(m-pn+pm)+O\big(\sqrt{n\log n}\big).\tag{by~\eqref{eq:bernoullirandoment}}
    \end{align*}
    Combining these inequalities yields the lemma statement.
\end{proof}

Recalling our choice of $\alpha$, Lemma~\ref{lem:fourthcond} tells us that the fourth condition of Lemma~\ref{lem:dtok} is satisfied if our choice of $n$ is sufficiently large for the $O\big(\sqrt{n\log n}\big)$ term to be at most $(n-m)(\alpha-(1-(1+p)\ent(p)/2))$. We now apply Lemma~\ref{lem:dtok} to prove the following.

\begin{lemma}\label{lem:nosuccess}
    $d_G^{(1-(\alpha+\varepsilon)/\gamma)}$ does not succeed on $S$.
\end{lemma}
\begin{proof}
    Let $\varepsilon\in(0,1)\cap\Q$, $\alpha$, and $\gamma$ be as above, and let $m_0$ be sufficiently large for Lemma~\ref{lem:dtok} to apply with $m=m_0$ and $n=\lceil m_0/\gamma\rceil$. For each $i\in\N$, let $m_{i+1}=\lceil m_i/\gamma\rceil$. Then by Lemma~\ref{lem:dtok}, for each $i\in\N$ we have
    \begin{equation}\label{eq:ratio}
        \max_{2m_i<k\leq 2m_{i+1}}\frac{d_G(S[0:k])}{d_G(S[0:2m_i])}\leq 2^{2(\alpha+\varepsilon)(m_{i+1}-m_i)}.
    \end{equation}
    For any $j\in\N$, let $n\in\N$ such that $2m_j\leq n\leq 2m_{j+1}$. Then,
    \begin{align*}
        d_G(S[0:n])&=d_G(S[0:2m_0])\left(\prod_{i=0}^{j-1}\frac{d_G(S[0:2m_{i+1}])}{d_G(S[0:2m_i])}\right)\frac{d_G(S[0:n])}{d_G(S[0:2m_j])}\\
        &=d_G(S[0:2m_0])\cdot 2^{2(\alpha+\varepsilon)m_{j+1}}\tag{by~\eqref{eq:ratio}}\\
        &\leq d_G(S[0:2m_0])\cdot 2^{(\alpha+\varepsilon)n/\gamma+2}.
    \end{align*}
    Therefore,
    \begin{align*}
        \limsup_{n\to\infty}d_G^{(1-(\alpha+\varepsilon)/\gamma)}(S[0:n])&=\limsup_{n\to\infty}2^{-(\alpha+\varepsilon)n/\gamma}d_G(S[0:n])\\
        &=\limsup_{n\to\infty} d_G(S[0:2m_0]),
    \end{align*}
    which is finite. That is, $d_G^{(1-(\alpha+\varepsilon)/\gamma)}$ does not succeed on $S$.
\end{proof}
    
We can choose $\varepsilon$ arbitrarily close to 0, $\alpha$ arbitrarily close to $1-(1+p)\ent(p)/2$, and $\gamma$ arbitrarily close to $1$, making $1-(\alpha+\varepsilon)/\gamma$ arbitrarily close to $(1+p)\ent(p)/2$. Although we only made this argument for rational values of $\alpha$, $\gamma$, and $\varepsilon$, success on a sequence is a monotone property in $s$, in that if $s'<s$ and $d_G^{(s')}$ succeeds on $S$, then $d_G^{(s)}$ also succeeds on $S$; this is immediate from the definition of success. Thus,
\[\inf\left\{s:d_G^{(s)}\text{ succeeds on }S\right\}\geq(1+p)\frac{\ent(p)}{2}.\]
As $G$ was an arbitrary oblivious multi-head finite-state gambler, we conclude that
\[\odimFS{\mh}(S)\geq (1+p)\frac{\ent(p)}{2}.\qedhere\]
    
\section{Main Theorem}\label{sec:main}

Given Theorems~\ref{thm:ub} and~\ref{thm:lb}, proving our separation theorem is simply a matter of choosing an appropriate bias $p$ for the underlying sequence $R$.
\begin{theorem}\label{thm:main}
    There is a sequence $S\in\{0,1\}^\omega$ such that
    \[\aDimFS{(2)}(S)<\odimFS{\mh}(S)-0.3.\]
\end{theorem}
\begin{proof}
    Let $R\in\{0,1\}^\omega$ be Martin-L\"of random with respect to the $\frac{\sqrt{2}}{2}$-biased Bernoulli measure $\mu_{\sqrt{2}/2}$, and let $S=F(R)$. Since $\frac{\sqrt{2}}{2}$ is computable, Theorem~\ref{thm:ub} applies, and since it is irrational, Theorem~\ref{thm:lb} applies. Therefore,
    \[\odimFS{\mh}(S)-\aDimFS{(2)}(S)\geq \frac{\sqrt{2}}{2}\frac{\ent(\sqrt{2}/2)}{2}> 0.30845.\qedhere\]
\end{proof}

While $p=\frac{\sqrt{2}}{2}$ has the advantage of simplicity, the maximum separation from Theorems~\ref{thm:ub} and~\ref{thm:lb}, namely, $\max_{p\in[0,1]}p\frac{\ent(p)}{2}$,
is slightly greater. It is open whether this is the maximum possible separation between adaptive two-head finite-state strong predimension and oblivious multi-head finite-state dimension.

Building on the hierarchy theorem of Huang et al.~\cite{mhfsd}, Cruz et al.~\cite{CGLL26} proved that for each $h\geq 2$ there is a sequence $Y\in\{0,1\}^\omega$ such that
\[\adimFS{(h-1)}(Y)\geq \oDimFS{(h)}(Y)+\frac{1}{p_h},\]
where $p_h$ is the $h$\textsuperscript{th} prime. Choosing $h=3$ and recalling~\eqref{eq:comparevariants}, this implies there is a sequence $Y\in\{0,1\}^\omega$ with
\[\aDimFS{(2)}(Y)\geq\odimFS{\mh}(Y)+0.2,\]
in sharp contrast to the sequence $S$ in Theorem~\ref{thm:main}. This points to a complex interplay between data access and sequence structure for multi-head gamblers, which merits further investigation.

\bibliographystyle{plain}
\bibliography{oath}

\end{document}